\documentclass{article}

\usepackage{hyperref}
\usepackage{amsmath,amssymb}
\usepackage[latin1]{inputenc}

\newcommand{\hide}[1]{}

\newcommand{\RR}{\mathbb{R}}

\newcommand{\NN}{\mathbb{N}}

\usepackage{comment,color}

\usepackage{fullpage}

\newtheorem{definition}{Definition}

\newtheorem{theorem}{Theorem}
\newtheorem{lemma}{Lemma}

\newenvironment{proof}{\paragraph{Proof:}}{\hfill$\square$}

\usepackage{algpseudocode}

\DeclareMathOperator*{\EE}{\mathbb{E}}

\begin{document}
	
	
	\title{Social Welfare in One-Sided Matching Mechanisms} 
%

\author{George Christodoulou\footnote{
		University of Liverpool, United Kingdom E-mail: {\tt gchristo@liv.ac.uk}}
	\and
	Aris Filos-Ratsikas\footnote{
		University of Oxford, United Kingdom. E-mail: {\tt aris.filos-ratsikas@cs.ox.ac.uk}}\\
	\and
	S\o ren Kristoffer Stiil Frederiksen\footnote{
		Aarhus University, Denmark. E-mail: {\tt sorensf@gmail.com}}
	\and
	Paul W. Goldberg \footnote{
		University of Oxford, United Kingdom. E-mail: {\tt paul.goldberg@cs.ox.ac.uk}}\\
	\and
	Jie Zhang \footnote{
		University of Oxford, United Kingdom. E-mail: {\tt jie.zhang@cs.ox.ac.uk}}\\
	\and
	Jinshan Zhang\footnote{
		University of Liverpool, United Kingdom E-mail: {\tt jinshan.zhang@liv.ac.uk}}
}
\date{} 
\maketitle
%

\begin{abstract}
We study the Price of Anarchy of mechanisms for the well-known problem
of one-sided matching, or house allocation, with respect to the social welfare objective.
We consider both ordinal mechanisms, where agents submit preference lists
over the items, and cardinal mechanisms, where agents may submit numerical
values for the items being allocated.
We present a general lower bound of $\Omega(\sqrt{n})$ on the Price of Anarchy,
which applies to \emph{all} mechanisms. We show that
two well-known mechanisms, Probabilistic Serial, and Random Priority,
achieve a matching upper bound.
We extend our lower bound to the Price of Stability of a large class of
mechanisms that satisfy a common proportionality property, and
show stronger bounds on the Price of Anarchy of all {\em deterministic} mechanisms.
\end{abstract}




\section{Introduction}

\emph{One-sided matching} (also called the house allocation problem)
is the fundamental problem of assigning
items to agents, such that each agent receives exactly one item.
It has numerous applications, such as assigning workers to
shifts, students to courses or patients to doctor appointments.
In this setting, agents are often asked to provide {\em ordinal
preferences}, i.e. preference lists, or rankings
of the items.
We assume that underlying these ordinal preferences,
agents have numerical values specifying how much they
value each item \cite{HZ:79}.
In game-theoretic terms, these are the agents' von Neumann-Morgenstern
utility functions~\cite{von2007theory} and the associated preferences
are often referred to as \emph{cardinal preferences}.

A \emph{mechanism} is a function that maps agents' valuations to
matchings. However, agents are rational strategic entities that might not always
report their valuations truthfully; they may misreport their values if
that results in a better matching (from their own perspective).
Assuming the agents report their valuations strategically to maximize
their utilities, it is of interest to study the
\emph{Nash equilibria} of the induced game, i.e. strategy profiles
from which no agent wishes to unilaterally deviate.
 
A natural objective for the designer is to choose the matching that
maximizes the \emph{social welfare}, i.e. the sum of agents'
valuations for the items they are matched with, which is the most prominent
measure of aggregate utility in the literature. Given the strategic nature
of the agents, we are interested in mechanisms that maximize the
social welfare \emph{in the equilibrium}. We use the standard measure
of equilibrium inefficiency, the \emph{Price of Anarchy}
\cite{koutsoupias1999worst}, that compares the maximum social welfare
attainable in any matching with the {\em worst-case} social welfare that can be achieved in
an equilibrium.



We evaluate the efficiency of a mechanism with respect to the Price of
Anarchy of the induced game.
We study both deterministic and randomized mechanisms: in the latter case 
the output is a probability mixture over matchings, instead of a
single matching. We are interested in the class of \emph{cardinal}
mechanisms, which use cardinal preferences, and generalize the ordinal
mechanisms.

Note that our setting involves no monetary
transfers and generally falls under the umbrella of
\emph{approximate mechanism design without money} \cite{PT:09}.
In general settings without money, one has to
fix a canonical representation of the valuations. A common
approach in the literature is to consider the \emph{unit-sum}
normalization, i.e. each agent has a total value of $1$ for all the items.
We obtain results for unit-sum valuations, and extend most of these
to another common normalization, \emph{unit-range}.

\subsection{Our results}

In Section~\ref{sec:PS} we bound the inefficiency of
the two best-known mechanisms in the
matching literature, \emph{Probabilistic Serial} and \emph{Random Priority}.
In particular, for $n$ agents and $n$ items, the Price of Anarchy is $O(\sqrt{n})$.
In Section~\ref{sec:Lower} we complement this with a \emph{matching} lower
bound (i.e. $\Omega(\sqrt{n})$) that applies to {\em all} cardinal (randomized) mechanisms.
As a result, we conclude that these two {\em ordinal} mechanisms
(ones that compute matchings that only depend on preference orderings)
are optimal. These results suggest that it does not help a welfare maximizer
to ask agents to report more than the ordinal preferences.

We separately consider \emph{deterministic} mechanisms
and in Section~\ref{sec:Lower} prove that their Price of Anarchy is $\Omega(n^2)$,
even for cardinal mechanisms. This shows that randomization is necessary
for non-trivial worst-case efficiency guarantees.

In Section \ref{app:solconcepts}, we extend our results to more general
solutions concepts as well as the case of incomplete information. 
Finally, in Section \ref{sec:POS}, we prove that under a mild ``proportionality''
property, our lower bound of $\Omega(\sqrt{n})$ extends to the \emph{Price of Stability},
a more optimistic measure of efficiency \cite{anshelevich2008price},
which strengthens the negative results even further.
Additionally, we discuss how our results extend to the other common normalization
in the literature, \emph{unit-range} \cite{qiang,FFZ:14,ZHOU:90}.



\subsection{Discussion and related work}

The one-sided matching problem was introduced in \cite{HZ:79} and has been
studied extensively ever since (see \cite{AS:13} for a recent overview).
Over the years, several different mechanisms have been proposed with various desirable
properties related to truthfulness, fairness and economic
efficiency with Probabilistic Serial \cite{BM:01,aziz2014strategic,BCK:11,AS:13}
and Random Priority \cite{AS:13,BM:01,aziz2013computational,krysta2014size,FFZ:14,qiang}
being the two prominent examples.

As mentioned earlier, in settings without money,
one needs to represent the valuations in some canonical way. 
A common approach is the \emph{unit-sum}
normalization, i.e. each agent has a total value of $1$ for all the
items. Intuitively, this normalization means that each agent has equal
influence within the mechanism and her values can be
interpreted as ``scrip money" that she uses to acquire items.
The unit-sum representation is standard for social welfare
maximization in many settings without money including fair division,
cake cutting and resource allocation
\cite{brams2012maxsum,caragiannis2012efficiency,GC:10,FFZ:14} among others.
Moreover, without any normalization, non-trivial Price of Anarchy bounds
cannot be achieved by any mechanism.

The objective of social welfare maximization for one-sided matching problems
has been studied before in the literature, but mainly for truthful mechanisms
 \cite{qiang,FFZ:14}. Our lower bounds are more general, since they apply to \emph{all} mechanisms, not just truthful ones. In particular, our lower bound on the Price of Anarchy of all mechanisms generalizes the corresponding bound for truthful mechanisms in \cite{FFZ:14}. Note that Random Priority is
truthful (truth-telling is a dominant strategy equilibrium) but it
has other equilibria as well; we observe that the welfare guarantees of the mechanism hold for
all equilibria, not just the truthtelling
ones. Similar approaches have been made for truthful mechanisms like
the second price auction in settings with money.

While given our general lower bound, proving a matching upper bound for Random Priority is enough to establish tightness, it is still important to know what the welfare guarantees of Probabilistic Serial are, given that it is arguably the most popular one-sided matching mechanism. The mechanism was introduced by \cite{BM:01} and since then, it has been in the center of attention of the matching literature, with related work on characterizations \cite{hashimoto2014two,kesten2006probabilistic}, extensions \cite{katta2006solution}, strategic aspects \cite{kojima2010incentives} and hardness of manipulation \cite{aziz2015manipulating}. Somewhat surprisingly, the Nash equilibria of the mechanism were only recently studied. Aziz et al. \cite{aziz2015equilibria} prove that the mechanism has pure Nash equilibria while Ekici and Kesten \cite{ekici2010ordinal} study the \emph{ordinal} equilibria of the mechanism and prove that the desirable properties of the mechanism are not necessarily satisfied for those profiles.

Another, somewhat different recent branch of study considers ordinal
measures of efficiency instead of social welfare maximization, under the assumption that agents'
preferences are only expressed through preference orderings over items.
Bhalgat et al.~\cite{BCK:11} study the approximation ratio of matching mechanisms,
when the objective is maximization of \emph{ordinal social welfare},
a notion of efficiency that they define based solely on ordinal information.
Other measures of efficiency for one-sided matchings were also studied in Krysta et al.~\cite{krysta2014size}, where the authors design truthful mechanisms to approximate the size of a maximum cardinally (or maximum agent weight) Pareto-optimal matching and in Chakrabarty and Swamy \cite{CS14} where the authors consider the rank approximation as the measure of efficiency. While interesting, these measures of efficiency
do not accurately encapsulate the socially desired outcome the way
that social welfare does, especially since an underlying cardinal
utility structure is part of the setting \cite{BM:01,HZ:79,von2007theory,ZHOU:90}.
Our results actually suggest that in order to achieve the optimal welfare
guarantees, one does not even need to elicit this utility structure;
agents can only be asked to report preference orderings, which is
arguably more appealing.

Finally, we point out that our work is in a sense analogous to the literature that studies the Price of Anarchy in item-bidding auctions (e.g. see \cite{CKS08,ST13} and references therein) for settings without money. Furthermore, the extension of our results to very general solution concepts (coarse correlated equilibria) and settings of incomplete information (Bayes-Nash equilibria) is somehow reminiscent of the \emph{smoothness} framework \cite{roughgarden2009intrinsic} for games. While our results are not proven using the smoothness condition, our extension technique is similar in spirit.



\section{Preliminaries} \label{sec:prelim}

Let $N=\{1,\ldots,n\}$ be a finite set of agents and $A=\{1,\ldots,n\}$ be a finite set of indivisible items. An \emph{allocation} is a matching of agents to items, that is, an assignment of items to agents where each agent gets assigned exactly one item. We can view an allocation $\mu$ as a permutation vector $(\mu_1,\mu_2\ldots, \mu_n)$ where $\mu_i$ is the unique item matched with agent $i$. Let $O$ be the set of all allocations. Each agent $i$ has a valuation function $u_i:A \rightarrow \RR$ mapping items to real numbers. Valuation functions are considered to be well-defined modulo positive affine transformations, that is, for item $j: j\rightarrow\alpha u_i(j)+\beta$ is considered to be an alternative representation of the same valuation function $u_i$. Given this, we fix the canonical representation of $u_i$ to be \emph{unit-sum}, that is $\sum_j u_i(j)=1$, with $u_i(j)\geq 0$ for all $i$, $j$. Equivalently, we can consider valuation functions as \emph{valuation vectors} $u_i =(u_{i1},u_{i2},\ldots,u_{in})$ and let $V$ be the set of all valuation vectors of an agent. Let $\mathbf{u}=(u_1,u_2,\ldots, u_n)$ denote a typical \emph{valuation profile} and let $V^n$ be the set of all valuation profiles with $n$ agents.

We consider \emph{strategic agents} who might have incentives to misreport their valuations.
We define $\mathbf{s}=(s_1,s_2,\ldots, s_n)$ to be a pure strategy profile, where $s_i$ is the \emph{reported} valuation vector of agent $i$. We will use $\mathbf{s_{-i}}$ to denote the strategy profile without the $i$th coordinate and hence $\mathbf{s}=(s_i,\mathbf{s_{-i}})$ is an alternative way to denote a strategy profile. A \emph{direct revelation mechanism} without money is a function $M: V^n\rightarrow  O$ mapping \emph{reported} valuation profiles to matchings. For a randomized mechanism, we define $M$ to be a random map $M:V^n\rightarrow O$. Let $M_i(\mathbf{s})$ denote the restriction of the outcome of the mechanism to the $i$'th coordinate, which is the item assigned to agent $i$ by the mechanism. For randomized mechanisms, we let $p^{M,\mathbf{s}}_{ij}=\Pr[M_i(\mathbf{s})=j]$ and $p_i^{M,\mathbf{s}}=(p^{M,\mathbf{s}}_{i1},\ldots,p^{M,\mathbf{s}}_{in})$. When it is clear from the context, we drop one or both of the superscripts from the terms $p^{M,\mathbf{s}}_{ij}$. The utility of an agent from the outcome of a deterministic mechanism $M$ on input strategy profile $\mathbf{s}$ is simply $u_i(M_i(\mathbf{s}))$. For randomized mechanisms, an agent's utility is $\EE[u_i(M_i(\mathbf{s}))] = \sum_{j=1}^n p_{ij}^{M,\mathbf{s}} u_{ij}$. 

A subclass of mechanisms that are of particular interest is that of \emph{ordinal mechanisms}. Informally, ordinal mechanisms operate solely based on the \emph{ordering} of items induced by the valuation functions and not the actual numerical values themselves, while cardinal mechanisms take those numerical values into account. Formally, a mechanism $M$ is \emph{ordinal} if for any strategy profiles $\mathbf{s},\mathbf{s'}$ such that for all agents $i$ and for all items $j,\ell$, $s_{ij}<s_{i\ell} \Leftrightarrow s'_{ij}<s'_{i\ell}$, it holds that $M(\mathbf{s})=M(\mathbf{s'})$. A mechanism for which the above does not necessarily hold is \emph{cardinal}. Equivalently, the strategy space of ordinal mechanisms is the set of all permutations of $n$ items instead of the space of valuation functions $V^n$. A strategy $s_i$ of agent $i$ is a \emph{preference ordering} of items $(a_1,a_2,\ldots,a_n)$ where $a_\ell \succ a_k$ for $\ell < k$. We will write $j \succ_i j'$ to denote that agent $i$ prefers item $j$ to item $j'$ according to her true valuation function and $j \succ_{s_i} j'$ to denote that she prefers item $j$ to item $j'$ according to her strategy $s_i$. When it is clear from the context, we abuse the notation slightly and let $u_i$ denote the truthtelling strategy of agent $i$, even when the mechanism is ordinal. Note that agents can be indifferent between items and hence the preference order can be a weak ordering.

Two properties of interest are \emph{anonymity} and \emph{neutrality}. A mechanism is anonymous if the output is invariant under renamings of the agents and neutral if the output is invariant under relabeling of the objects. 

An \emph{equilibrium} is a strategy profile in which no agent has an incentive to deviate to a different strategy. First, we will focus on the concept of \emph{pure Nash equilibrium}, formally 
\begin{definition}
A strategy profile $\mathbf{s}$ is a \emph{pure Nash equilibrium} if $u_i(M_i(\mathbf{s})) \geq u_i(M_i(s'_i,s_{-i}))$ for all agents $i$, and pure deviating strategies $s'_i$.
\end{definition}
In Section \ref{app:solconcepts}, we extend our results to more general equilibrium notions as well as the setting of incomplete information, where agents' values are drawn from known distributions. 
Let $S_{\mathbf{u}}^M$ denote the set of all pure Nash equilibria of mechanism $M$ under truthful valuation profile $\mathbf{u}$. The measure of efficiency that we will use is the \emph{pure Price of Anarchy},
\begin{equation*}
PoA(M) = \sup_{\mathbf{u} \in V^n} \frac{SW_{OPT}(\mathbf{u})}{\min_{\mathbf s \in S_\mathbf{u}^M}SW_M(\mathbf{u,s})}
\end{equation*}
\noindent where $SW_M(\mathbf{u,s})=\sum_{i=1}^{n}\EE[u_i(M_i(\mathbf{s}))]$ is  the expected \emph{social welfare} of mechanism $M$ on strategy profile $\mathbf{s}$ under true valuation profile $\mathbf{u}$, and $SW_{OPT}(\mathbf{u})=\max_{\mu \in O}\sum_{i=1}^{n}u_i(\mu_i)$ is the social welfare of the optimal matching. Let $OPT(\mathbf{u})$ be the optimal matching on profile $\mathbf{u}$ and let $OPT_i(\mathbf{u})$ be the restriction to the $i$th coordinate. 


\section{Price of Anarchy guarantees}\label{sec:PS}

In this section, we prove the (pure) Price of Anarchy guarantees of Probabilistic Serial and Random Priority. Together with our lower bound in the next section, the results establish that both mechanisms are optimal among all mechanisms for the problem. 

First we consider Random Priority, often referred to as Random Serial Dictatorship. The mechanism first fixes an ordering of the agents uniformly at random and then according to that ordering, it sequentially matches them with their most preferred item that is still available. Filos-Ratsikas et al.~\cite{FFZ:14} proved that the welfare in any truthtelling equilibrium is an $\Omega(1/\sqrt{n})$-fraction of the maximum social welfare. While Random Priority has other equilibria as well, to establish the Price of Anarchy bound, it suffices to observe that at least for distinct valuations, any strategy other than truthtelling does not affect the allocation and hence it does not affect the social welfare. Intuitively, since agents pick their most preferred items, any equilibrium strategy would place the most preferred available items on top of the preference list, while the ordering of the items that are not picked does not affect the allocation of other agents. For valuations that are not distinct, the argument can be adapted using small perturbations of the values, losing only a small fraction of welfare. We state the theorem; the details are left for the full version.  

\begin{theorem}
The Price of Anarchy of Random Priority is $O(\sqrt{n})$.
\end{theorem}
We next consider \emph{Probabilistic Serial}, which we abbreviate to $PS$.
Informally, the mechanism is the following. Each item can be viewed as an infinitely divisible good that all agents can consume at unit speed during the unit time interval $[0,1]$. Initially each agent consumes her most preferred item (or one of her most preferred items in case of ties) until the item is entirely consumed. Then, the agent moves on to consume the item on top of her preference list, among items that have not yet been entirely consumed. The mechanism terminates when all items have been entirely consumed. The fraction $p_{ij}$ of item $j$ consumed by agent $i$ is then interpreted as the probability that agent $i$ will be matched with item $j$ under the mechanism.

We prove that the Price of Anarchy of $PS$ is $O(\sqrt{n})$.
Aziz et al.~\cite{aziz2015equilibria} proved that $PS$ has pure Nash equilibria,
so it makes sense to consider the pure Price of Anarchy; we will extend the result to the coarse correlated Price of Anarchy and the Bayesian Price of Anarchy in Section \ref{app:solconcepts}.

We start with the following two lemmas, which prove that in a
pure Nash equilibrium of the mechanism an agent's utility cannot
be much worse than what her utility would be if she were consuming
the item she is matched with in the optimal allocation from the
beginning of the mechanism until the item is entirely consumed.
Let $t_j(\mathbf{s})$ be the time when item $j$ is entirely consumed
on profile $\mathbf{s}$ under $PS(\mathbf{s})$.

\begin{lemma}\label{lem:time}
Let $\mathbf{s}$ be any strategy profile and let $s_i^*$ be any strategy such that $j \succ_{s_i^*} \ell$ for all items $\ell \neq j$, i.e. agent $i$ places item $j$ on top of her preference list. Then it holds that $t_j(s_i^*,\mathbf{s_{-i}})  \geq \frac{1}{4} \cdot t_j(\mathbf{s})$.   
\end{lemma}

\begin{proof}
	For ease of notation, let $\mathbf{s^*}=(s_i^*,\mathbf{s_{-i}})$. Obviously, if $j \succ_{s_i} \ell$ for all $\ell \neq j$ and since all other agents' reports are fixed, $t_j(\mathbf{s^*})=t_j(\mathbf{s})$ and the statement of the lemma holds. Hence, we will assume that there exists some item $j' \neq j$ such that $j' \succ_{s_i} j$.  
	
	First, note that if agent $i$ is the only one consuming item $j$ for the duration of the mechanism, then $t_j(\mathbf{s^*})=1$ and we are done. Hence, assume that at least one other agent consumes item $j$ at some point, and let $\tau$ be the time when the first agent besides agent $i$ starts consuming item $j$ in $\mathbf{s^*}$. Obviously, $t_j(\mathbf{s^*}) > \tau$, therefore if $\tau \geq \frac{1}{4} \cdot t_j(\mathbf{s})$ then $t_j(\mathbf{s^*}) \geq \frac{1}{4}\cdot t_j(\mathbf{s}) $ and we are done. So assume that $\tau < \frac{1}{4} \cdot t_j(\mathbf{s})$. Next observe that in the interval $[\tau,t_j(\mathbf{s^*})]$, agent $i$ can consume at most half of what remains of item $i$ because there exists at least one other agent consuming the item for the same duration. Overall, agent $i$'s consumption is at most $\frac{1}{2}+\frac{1}{4}t_j(\mathbf{s})$ so at least $\frac{1}{2}-\frac{1}{4}t_j(\mathbf{s})$ of the item will be consumed by the rest of the agents.
	
	Now consider all agents other than $i$ in profile $\mathbf{\mathbf{s}}$ and let $\alpha$ be the the amount of item $j$ that they have consumed by time $t_j(\mathbf{s})$. Notice that the total consumption speed of an item is non-decreasing in time which means in particular that for any $0 \leq \beta \leq 1$, agents other than $i$ need at least $\beta t_j(\mathbf{s})$ time to consume $\alpha \cdot \beta$ in profile $\mathbf{s}$. Next, notice that since agent $i$ starts consuming item $j$ at time $0$ in $\mathbf{\mathbf{s^*}}$ and all other agents use the same strategies in $\mathbf{s}$ and $\mathbf{s^*}$, it holds that every agent $k \neq i$ starts consuming item $j$ in $\mathbf{s^*}$ no sooner than she does in $\mathbf{s}$. This means that in profile $\mathbf{s^*}$, agents other than $i$ will need more time to consume $\beta \cdot \alpha$; in particular they will need at least $\beta t_j(\mathbf{s})$ time, so $t_j(\mathbf{s^*}) \geq \beta t_j(\mathbf{s})$. However, from the previous paragraph we know that they will consume at least $\frac{1}{2}-\frac{1}{4}t_j(\mathbf{s})$, so letting $\beta = \frac{1}{\alpha}\left(\frac{1}{2} -  \frac{1}{4}t_j(\mathbf{s}) \right) $ we get
	\begin{eqnarray*}
		t_j(\mathbf{s^*}) \geq \beta t_j(\mathbf{s}) \geq t_j(\mathbf{s}) \left(\frac{1}{2} -  \frac{1}{4}\cdot t_j(\mathbf{s}) \right) \frac{1}{\alpha}  \\
		\geq t_j(\mathbf{s}) \left(\frac{1}{2} -  \frac{1}{4}\cdot t_j(\mathbf{s}) \right) \geq \frac{1}{4} \cdot t_j(\mathbf{s})
	\end{eqnarray*}
\end{proof}	
Now we can lower bound the utility of an agent at any pure Nash equilibrium.

\begin{lemma}\label{lem:tv}
Let $\mathbf{u}$ be the profile of true agent valuations and let $\mathbf{s}$ be a pure Nash equilibrium. For any agent $i$ and any item $j$ it holds that the utility of agent $i$ at $\mathbf{s}$ is at least $\frac{1}{4} \cdot t_{j}(\mathbf{s}) \cdot u_{ij}$.
\end{lemma}

\begin{proof}
Let $\mathbf{s'}=(s'_{i},\mathbf{s_{-i}})$ be the strategy profile obtained from $\mathbf{s}$ when agent $i$ deviates to the strategy $s_i'$ where $s_i'$ is some strategy such that $j \succ_{s_i'} \ell$ for all items $\ell \neq j$. Since $\mathbf{s}$ is a pure Nash equilibrium, it holds that $u_i(PS_i(\mathbf{s})) \geq u_i(PS_i(\mathbf{s'})) \geq t_j(\mathbf{s'})\cdot u_{ij}$, where the last inequality holds since the utility of agent $i$ is at least as much as the utility she obtains from the consumption of item $j$.  By Lemma \ref{lem:time}, it holds that $t_j(\mathbf{s'})  \geq \frac{1}{4} \cdot t_j(\mathbf{s})$ and hence $u_i(PS_i(\mathbf{s})) \geq \frac{1}{4} \cdot t_j(\mathbf{s})\cdot u_{ij}$.  
\end{proof}
We can now prove the pure Price of Anarchy guarantee of the mechanism.

\begin{theorem}\label{thm:PSPOA}
The pure Price of Anarchy of Probabilistic Serial is $O(\sqrt{n})$.
\end{theorem}

\begin{proof}
Let $\mathbf{u}$ be any profile of true agents' valuations and let $\mathbf{s}$ be any pure Nash equilibrium. First, note that by reporting truthfully, every agent $i$ can get an allocation that is at least as good as $\left(1/n,\ldots, 1/n\right)$, regardless of other agents' strategies. To see this, first consider time $t=1/n$ and observe that during the interval $[0,1/n]$, agent $i$ is consuming her favorite item (say $a_1$) and hence $p_{ia_1} \geq 1/n$. Next, consider time $\tau=2/n$ and observe that during the interval $[0,2/n]$, agent $i$ is consuming one or both of her two favorite items ($a_1$ and $a_2$) and hence $p_{ia_1}+p_{ia_2} \geq 2/n$. By a similar argument, for any $k$, it holds that $\sum_{j=1}^np_{ia_j} \geq k/n$. This implies that regardless of other agents' strategies, agent $i$ can achieve a utility of at least $\frac{1}{n}\sum_{j=1}^{n} u_{ij}$. Since $\mathbf{s}$ is a pure Nash equilibrium, it holds that $u_i(PS_i(\mathbf{s})) \geq (1/n)\sum_{j=1}^{n} u_{ij}$ as well. Summing over all agents, we get that $SW_{PS}(\mathbf{u,s}) \geq (1/n) \sum_{i=1}^n \sum_{j=1}^n u_{ij}  = 1$. If $SW_{OPT}(\mathbf{u}) \leq \sqrt{n}$ then we are done, so assume $SW_{OPT}(\mathbf{u}) > \sqrt{n}$.
	
	Because $PS$ is neutral we can assume $t_j(\mathbf{s}) \leq t_{j'}(\mathbf{s})$ for $j<j'$ without loss of generality. Observe that for all $j=1,\ldots,n$, it holds that $t_j(\mathbf{s})  \geq j/n$. This is true because for any $t
	\in [0,1]$, by time $t$, exactly $tn$ mass of items must have been consumed by the agents. Since $j$ is the $j$th item that is entirely consumed, by time $t_j(\mathbf{s}) $, the mass of items that must have been consumed is at least $j$. By this, we get that $t_j(\mathbf{s}) \cdot  n \geq j$, which implies $t_j(\mathbf{s})  \geq j/n$. 
	
	For each $j$ let $i_j$ be the agent that gets item $j$ in the optimal allocation and for ease of notation, let $w_{i_j}$ be her valuation for the item. Now by Lemma \ref{lem:tv}, it holds that \[u_{i_j}(PS(\mathbf{s})) \geq \frac{1}{4} \frac{j}{n} w_{i_j} \ \ \text{ and } \ \ SW_{PS}(\mathbf{u,s}) \geq \frac{1}{4}\sum_{j=1}^{n}  \frac{j}{n} w_{i_j}.\] The Price of Anarchy is then at most \[\frac{4\sum_{j=1}^{n}w_{i_j}}{\sum_{j=1}^{n} j\cdot w_{i_j} /n}.\]
	Consider the case when the above ratio is maximized and let $k$ be an integer such that $k\le\sum_{j=1}^{n} w_{i_j} \le k+1$. Then it must be that $w_{i_j}=1$ for $j=1,\ldots,k$ and $w_{i_j}=0$, for $k+2\leq i_j\leq n$. Hence the maximum ratio is $(k+w_{i_{k+1}})/(aw_{i_{k+1}}+b)$, for some $a$, $b>0$, which is monotone for $w_{i_{k+1}}$ in $[0,1]$. Therefore, the maximum value of $(k+w_{i_{k+1}})/(aw_{i_{k+1}}+b)$ is achieved when either $w_{i_{k+1}}=0$ or $w_{i_{k+1}}=1$. As a result, the maximum value of the ratio is obtained when $\sum_{i=1^n}w_{i_{k+1}}=k$ for some $k$. By simple calculations, the Price of Anarchy should be at most:
	\begin{eqnarray*}
		\frac{4k}{\sum_{j=1}^{k} \frac{j}{n}} \leq \frac{4k}{ \frac{k(k-1)}{2n}} = \frac{8n}{k-1},
	\end{eqnarray*}
	so the Price of Anarchy is maximized when $k$ is minimized. By the argument earlier, $k > \sqrt{n}$ and hence the ratio is $O(\sqrt{n})$.
\end{proof}
In Section \ref{app:solconcepts}, we extend Theorem \ref{thm:PSPOA} to
broader solution concepts and the incomplete information setting.
\section{Lower bounds}\label{sec:Lower}

In this section, we prove our main lower bound. Note that the result holds for any mechanism, including randomized and cardinal mechanisms. Since we are interested in mechanisms with good properties, it is natural to consider those mechanisms that have pure Nash equilibria.

\begin{theorem}\label{thm:loweranyunitsum}
	The pure Price of Anarchy of any mechanism is $\Omega(\sqrt{n})$.
\end{theorem}

\begin{proof}
	Assume $n=k^2$ for some $k \in \NN$. Let $M$ be a mechanism and consider the following valuation profile $\mathbf{u}$. There are $\sqrt{n}$ sets of agents and let $G_j$ denote the $j$-th set. For every $j \in \{1,\ldots,\sqrt{n}\}$ and every agent $i \in G_j$, it holds that $u_{ij}=1/n+\alpha$ and $u_{ik} = 1/n-\alpha/(n-1)$, for $k\neq j$, where $\alpha$ is sufficiently small. Let $\mathbf{s}$ be a pure Nash equilibrium and for every set $G_j$, let $i_j = \arg \min_{i \in G_j} p_{ij}^{M,\mathbf{s}}$ (break ties arbitrarily). Observe that for all $j=1,\ldots,\sqrt{n}$, it holds that $p_{i_{j}j}^{M,\mathbf{s}} \leq 1/\sqrt{n}$ and let $I=\{i_1,i_2,\ldots,i_{\sqrt{n}}\}$. Now consider the valuation profile $\mathbf{u'}$ where:
	\begin{itemize}
		\item For every agent $i \notin I$, $u_i' = u_i$.
		\item For every agent $i_j \in I$, let $u_{i_j j}' = 1$ and $u_{i_jk}' = 0$ for all $k \neq j$.
	\end{itemize}
	We claim that $\mathbf{s}$ is a pure Nash equilibrium under $\mathbf{u'}$ as well. For agents not in $I$, the valuations have not changed and hence they have no incentive to deviate. Assume now for contradiction that some agent $i \in I$ whose most preferred item is item $j$ could deviate to some beneficial strategy $s_i'$. Since agent $i$ only values item $j$, this would imply that $p_{ij}^{M,(s_i',\mathbf{s_{-i})}} > p_{ij}^{M,\mathbf{s}}$.	However, since agent $i$ values all items other than $j$ equally under $u_i$ and her most preferred item is item $j$, such a deviation would also be beneficial under profile $\mathbf{u}$, contradicting the fact that $\mathbf{s}$ is a pure Nash equilibrium.
	
	Now consider the expected social welfare of $M$ under valuation profile $\mathbf{u'}$ at the pure Nash equilibrium $\mathbf{s}$. For agents not in $I$ and taking $\alpha$ to be less than $1/n^3$, the contribution to the social welfare is at most $1$. For agents in $I$, the contribution to the welfare is then at most $(1/\sqrt{n}) \sqrt{n} + 1$ and hence the expected social welfare of $M$ is at most $3$. As the optimal social welfare is at least $\sqrt{n}$, the bound follows.
\end{proof}
Interestingly, if we restrict our attention to {\em deterministic} mechanisms,
then we can prove that only trivial pure Price of Anarchy guarantees are achievable.  

\begin{theorem}\label{thm:deterpoa}
	The pure Price of Anarchy of any deterministic mechanism is $\Omega(n^2)$.
\end{theorem}

\begin{proof}
Let $M$ be a deterministic mechanism that always has a pure Nash equilibrium.
Let $\mathbf{u}$ be a valuation profile such that for for all agents $i$ and $i'$,
it holds that $u_i=u_{i'}$, $u_{i1}=1/n+ 1/n^3$ and $u_{ij} > u_{ik}$ for $j < k$.
Let $\mathbf{s}$ be a pure Nash equilibrium for this profile and assume without loss of generality that $M_i(\mathbf{s})=i$. 
	
	Now fix another true valuation profile $\mathbf{u'}$ such that $u'_1=u_1$ and for agents $i=2,\ldots,n$, $u'_{i,i-1} = 1-\epsilon'_{i,i-1}$ and $u_{ij} = \epsilon'_{ij}$ for $j\neq i-1$, where $0 \leq \epsilon'_{ij} \leq 1/n^3$, $\sum_{j \neq i-1} \epsilon'_{ij} = \epsilon'_{i,i-1}$ and $\epsilon'_{ij} > \epsilon'_{ik}$ if $j < k$ when $j,k \neq i-1$. Intuitively, in profile $\mathbf{u'}$, each agent $i \in \{2,\ldots,n\}$ has valuation close to $1$ for item $i-1$ and small valuations for all other items. Futhermore, she prefers items with smaller indices, except for item $i-1$.
	
	We claim that $\mathbf{s}$ is a pure Nash equilibrium under true valuation profile $\mathbf{u}$ as well. Assume for contradiction that some agent $i$ has a benefiting deviation, which matches her with an item that she prefers more than $i$. But then, since the set of items that she prefers more than $i$ in both $\mathbf{u}$ and $\mathbf{u'}$ is $\{1,\ldots,i\}$, the same deviation would match her with a more preferred item under $\mathbf{u}$ as well, contradicting the fact that $\mathbf{s}$ is a pure Nash equilibrium. It holds that $SW_{OPT}(\mathbf{u'})\geq n-2$ whereas the social welfare of $M$ is at most $2/n$ and the theorem follows.
\end{proof}
The mechanism that naively maximizes the sum of
the reported valuations with no regard to incentives, when equipped with
a lexicographic tie-breaking rule has pure Nash equilibria and also
achieves the above ratio in the worst-case, which means that the bounds are tight.

\section{General solution concepts}\label{app:solconcepts}

In the previous sections, we employed the pure Nash equilibrium
as the solution concept for bounding the inefficiency of mechanisms,
mainly because of its simplicity.
Here, we describe how to extend our results to broader well-known equilibrium
concepts in the literature. Due to lack of space, we will only discuss the two most general solution concepts, the \emph{coarse correlated equilibrium} for complete information and the \emph{Bayes-Nash equilibrium} for incomplete information. Since other concepts (like the mixed-Nash equilibrium for instance) are special cases of those two, it suffices to use those for our extensions.

 \begin{definition}
Given a mechanism $M$, let $\mathbf{q}$ be a distribution over strategies. Also, for any distribution $\Delta$ let $\Delta_{-i}$ denote the marginal distribution without the $i$th index. Then a strategy profile $\mathbf{q}$ is called a
\begin{enumerate}
\item \emph{coarse correlated equilibrium} if \[\EE_{\mathbf{s}\sim \mathbf{q}}[u_i(M_i(\mathbf{s}))]\geq \EE_{\mathbf{s}\sim \mathbf{q}}[u_i(M_i((s'_i,\mathbf{s_{-i}})))],\]
\item \emph{Bayes-Nash equilibrium} for a distribution $\Delta_u$ where each $(\Delta_u)_i$ is independent, if when $\mathbf{u} \sim \Delta_u$ then $\mathbf{q(u)}=\times_{i}q_i(u_i)$ and for all $u_i$ in the support of $(\Delta_u)_i$,
\begin{small} \[\EE_{\mathbf{u_{-i}},\mathbf{s}\sim 
	\mathbf{q(u)}}[u_i(M_i(\mathbf{s}))]\geq \EE_{\mathbf{u_{-i}, s_{-i}}\sim \mathbf{q_{-i}(u_{-i})}}[u_i(M_i(s'_i,\mathbf{s_{-i}}))]\]\end{small}
\end{enumerate}
where the given inequalities hold for all agents $i$, and (pure) deviating strategies $s'_i$. Also notice that for randomized mechanisms definitions are with respect to an expectation over the random choices of the mechanism. 
\end{definition}
The coarse correlated and the Bayesian Price of Anarchy are defined similarly to the pure Price of Anarchy. 

Again, first we mention that we can obtain the extensions to Random Priority rather straightforwardly, based on the fact that even when using probability mixtures over strategies, an agent will always (in every realization) pick her most preferred item among the set of available items when she is chosen. In other words, any pure strategy in the support of the distribution will rank the most preferred available item first, and the ordering of the remaining items does not affect the distribution. Since the arguments are not very involved, we leave the details for the full version.

\begin{theorem}
	The coarse correlated Price of Anarchy of Random Priority is $O(\sqrt{n})$. The Bayesian Price of Anarchy of Random Priority is $O(\sqrt{n})$.
\end{theorem} 
Next, we turn to Probabilistic Serial and prove the Price of Anarchy guarantees, with respect to coarse correlated equilibria and Bayes-Nash equilibria. Before we state our theorems however, we will briefly discuss the connection of those extensions with the \emph{smoothness} framework of Roughgarden \cite{roughgarden2009intrinsic}. According to the definition in \cite{roughgarden2009intrinsic}, a game is $(\lambda,\mu)$-\emph{smooth} if it satisfies the following condition
\begin{equation} \label{smoothness}
\sum_{i=1}^{n} u_i (s_i^*,\mathbf{s_{-i}}) \geq \lambda SW(\mathbf{s^*}) - \mu SW(\mathbf{s}), 
\end{equation}
where $\mathbf{s^*}$ is a pure strategy profile that corresponds to the optimal allocation and $\mathbf{s}$ is any pure strategy profile. It is not hard to see that a $(\lambda,\mu)$-smooth game has a Price of Anarchy bounded by $(\mu+1)/\lambda$. 

Since establishing that a game is smooth also implies a pure Price of Anarchy bound, an alternative way of attempting to prove Theorem \ref{thm:PSPOA} would be to try to show smoothness of the game induced by PS, for $ \mu/\lambda=\sqrt{n}$. However, this seems to be a harder task than what we actually do, since in such a proof, one would have to argue about the utilities of agents and possibly reason about the relative preferences for other items, other than the item they are matched with in the optimal allocation. Our approach only needs to consider those items, and hence it seems to be simpler.

An added benefit to the smoothness framework is the existence of the \emph{extension theorem} in \cite{roughgarden2009intrinsic}, which states that for a $(\lambda,\mu)$-smooth game, the Price of Anarchy guarantee extends to broader solution concepts verbatim, without any extra work. At first glance, one might think that proving smoothness for the game induced by PS might be worth the extra effort, since we would get the extensions ``for free''. A closer look at our proofs however shows that our approach is very similar to the proof of the extension theorem but using an alternative, simpler condition.  

Specifically, the analysis in \cite{roughgarden2009intrinsic} uses Inequality \ref{smoothness} as a building block and substitutes the inequality into the expectations that naturally appear when considering randomized strategies. This can be done because the condition applies to all strategy profiles $\mathbf{s}$, when $\mathbf{s^*}$ is an optimal strategy profile. This is exactly what we do as well, but we use the inequality $t_j(s_i^*,\mathbf{s_{-i}})  \geq \frac{1}{4} \cdot t_j(\mathbf{s})$ instead, which is simpler but sufficient since it only applies to the game at hand. If $OPT_i(\mathbf{u})=j$, which is what we use in the proof of Theorem \ref{thm:PSPOA}, then $(s_i^*,\mathbf{s_{-i}})$ can be thought of as a profile where an agent deviates to her strategy in the optimal profile and hence the left-hand side of the inequality is analogous to  the left-hand side of Inequality \ref{smoothness}. In a sense, the inequality $t_j(s_i^*,\mathbf{s_{-i}})  \geq \frac{1}{4} \cdot t_j(\mathbf{s})$, can be viewed as a ``smoothness equivalent'' for the game induced by PS, which then allows us to extend the results to broader solution concepts.       

First, we extend Theorem \ref{thm:PSPOA} to the case where the solution concept is the coarse correlated equilibrium. 

\begin{theorem}\label{PSPOAcoarse}
	The coarse correlated Price of Anarchy of Probabilistic Serial is $O(\sqrt{n})$.
\end{theorem}

\begin{proof}
	Let $\mathbf{u}$ be any valuation profile and let $i$ be any agent. Furthermore, let $j = OPT_i(\mathbf{u})$ and let $s'_i$ be the pure strategy that places item $j$ on top of agent $i$'s preference list. By Lemma \ref{lem:time}, the inequality $t_j(s_i',\mathbf{s_{-i}}) \geq \frac{1}{4} t_j(\mathbf{s})$ holds for every strategy profile $\mathbf{s}$. In particular, it holds for any pure strategy profile $\mathbf{s}$ where $s_i$ is in the support of the distribution of the mixed strategy $q_i$ of agent $i$, for any coarse correlated equilibrium $\mathbf{q}$.This implies that
	\begin{eqnarray*}\label{mixed-truth}
		\EE_{\mathbf{s}\sim \mathbf{q}}[u_i(PS_i(\mathbf{s}))]&\ge& \EE_{\mathbf{s}\sim \mathbf{q}}[u_i(PS_i(s'_i,\mathbf{s_{-i}}))]
		\\&\ge&\EE_{\mathbf{s}\sim \mathbf{q}}[u_{ij}t_j(s'_i,\mathbf{s_{-i}}))]\ge  \frac{1}{4}u_{ij} t_j(\mathbf{s}).
	\end{eqnarray*}
	where the last inequality holds by Lemma \ref{lem:time}. Using this, we can use very similar arguments to the arguments of the proof of Theorem \ref{thm:PSPOA} and obtain the bound.
\end{proof}
For the incomplete information setting, when valuations are drawn from some publically known distributions, we can prove the same upper bound on the Bayesian Price of Anarchy of the mechanism.

\begin{theorem}\label{PS:POABayesian}
	The Bayesian Price of Anarchy of Probabilistic Serial is $O(\sqrt{n})$.
\end{theorem}

\begin{proof}
	The proof is again similar to the proof of Theorem \ref{thm:PSPOA}. Let $\mathbf{u}$ be a valuation profile drawn from some distribution satisfying the unit-sum constraint. Let $i$ be any agent and let $j_u = OPT_i(\mathbf{u})$, $i\in [n]$. Note that by a similar argument as the one used in the proof of Theorem \ref{thm:PSPOA}, the expected social welfare of $PS$ is at least $1$ and hence we can assume that $\EE_{\mathbf{u}}[SW_{OPT}(\mathbf{u})] \geq 2\sqrt{2n}+1$. 
	Observe that in any Bayes-Nash equilibrium $\mathbf{q(u)}$ it holds that
	\begin{eqnarray*}\label{equ:bayesian-1}
	\EE_{\substack{\mathbf{u}\\\mathbf{s}\sim \mathbf{q(u)}} }\left[u_i(\mathbf{s})\right]
	&=&\EE_{u_i}\left[\EE_{\substack{\mathbf{u_{-i}}\\\mathbf{s}\sim \mathbf{q(u)}} }\left[u_i(\mathbf{s})\right]\right]\\&\ge&
	\EE_{u_i}\left[\EE_{\substack{\mathbf{u_{-i}}\\\mathbf{s_{-i}}\sim \mathbf{q_{-i}(u_{-i})}} }\left[u_i(s'_i,\mathbf{s_{-i}})\right]\right]\\ &\ge& \EE_{u_i}\left[\EE_{\substack{\mathbf{u_{-i}}\\\mathbf{s_{-i}}\sim \mathbf{q_{-i}(u_{-i})}}}\left[u_{ij_u}t_{j_u}(s'_i,\mathbf{s_{-i}})\right] \right]\\&\ge& \EE_{u_i}\left[\EE_{\substack{\mathbf{u_{-i}}\\\mathbf{s}\sim \mathbf{q(u)}}}\left[\frac{1}{4}u_{ij_u}t_{j_u}(\mathbf{s})\right] \right]\\
	&=&\frac{1}{4} \EE_{\substack{\mathbf{u}\\\mathbf{s}\sim \mathbf{q(u)}}}\left[u_{ij_u}t_{j_u}(\mathbf{s})\right]  
	\end{eqnarray*}
	where the last inequality holds by Lemma \ref{lem:time} since $s_i'$ denotes the strategy that puts item $j_u$ on top of agent $i$'s preference list. Note that this can be a different strategy for every different $\mathbf{u}$ that we sample. For notational convenience, we use $s_i'$ to denote every such strategy. The expected social welfare at the Bayes-Nash equilibrium is then at least
	\begin{small}
	\begin{eqnarray*}\label{equ:bayesian-2}
	\sum_{i=1}^n\EE_{\mathbf{u},\mathbf{s}\sim \mathbf{q(u)}}\left[u_i(\mathbf{s})\right]
	&\ge &\frac{1}{4}\sum_{i\in [n]}\EE_{\substack{\mathbf{u}\\\mathbf{s}\sim \mathbf{q(u)}}}\left[u_{ij_u}t_{j_u}(\mathbf{s})\right]
	\\&\ge& \EE_{\substack{\mathbf{u}\\\mathbf{s}\sim \mathbf{q(u)}}}\left[\sum_{i\in [n]} \frac{i}{4n}u_{ij_u} \right]\\
	&\ge& \EE_{\substack{\mathbf{u}\\\mathbf{s}\sim \mathbf{q(u)}}}\left[\frac{SW_{OPT}(\mathbf{u})(SW_{OPT}(\mathbf{u})-1)}{8n}\right]\\
    &=&\EE_{\mathbf{u}}\left[\frac{SW_{OPT}(\mathbf{u})(SW_{OPT}(\mathbf{u})-1)}{8n}\right]\\
	&\ge& \frac{\EE_{\mathbf{u}}\left[\left(SW_{OPT}(\mathbf{u})\right)^2\right]-\EE_{\mathbf{u}}\left[SW_{OPT}(\mathbf{u})\right]}{8n}\\
	&\ge& \frac{\EE_{\mathbf{u}}[SW_{OPT}(\mathbf{u})] }{2\sqrt{2n}},
	\end{eqnarray*}
\end{small}
	
	and the bound follows.
\end{proof}	\bigskip\bigskip

%
%
%
%
%
%

\section{Extensions}\label{sec:POS}

\subsection{Price of Stability}
Theorem \ref{thm:loweranyunitsum} bounds the Price of Anarchy of all mechanisms.
A more optimistic (and hence stronger when proving lower bounds) measure of efficiency is the \emph{Price of Stability},
i.e. the worst-case ratio over all valuation profiles of the optimal social
welfare over the welfare attained at the \emph{best} equilibrium.
We extend Theorem \ref{thm:loweranyunitsum}  to the Price of Stability of all mechanisms that satisfy a ``proportionality'' property.

Let $a_1 \succ_i a_2 \succ_i \cdots \succ_i a_n$ be the (possibly weak) preference ordering of agent $i$. A random assignment vector $p_i$ for agent $i$ \emph{stochastically dominates} another random assignment vector $q_i$ if $\sum_{j=1}^k p_{ia_j} \geq \sum_{j=1}^k q_{ia_j}$, for all $k=1,2,\cdots, n$.
	The notation that we will use for this relation is $p_i \succ_i^{sd} q_i$.

\begin{definition}[Safe strategy]
	Let $M$ be a mechanism. A strategy $s_i$ is a \emph{safe strategy} if for any strategy profile $s_{-i}$ of the other players, it holds that $M_i(s_i,s_{-i}) \succ_{i}^{sd} \left(\frac{1}{n}, \frac{1}{n}, \ldots, \frac{1}{n} \right)$.
\end{definition}
We will say that a mechanism $M$ has a safe strategy if every agent $i$ has a safe strategy $s_i$ in $M$. We now state our lower bound.

\begin{theorem}\label{thm-PoS-lower}
	The pure Price of Stability of any mechanism that has a safe strategy is $\Omega(\sqrt{n})$.
\end{theorem}

\begin{proof}
	Let $M$ be a mechanism and let $I= \{k+1,\ldots,n\}$ be a subset of agents. Let $\mathbf{u}$ be the following valuation profile.
\begin{itemize}
	\item For all agents $i \in I$, let $u_{ij} = \frac{1}{k}$ for $j = 1, \cdots, k$ and $u_{ij}=0$ otherwise.
	\item For all agents $i \notin I$, let $u_{ii}=1$ and $u_{ij}=0, j\neq i$.
\end{itemize}
	
	Now let $\mathbf{s}$ be a pure Nash equilibrium on profile $\mathbf{u}$ and let $s'_i$ be a safe strategy of agent $i$. The expected utility of each agent $i\in I$ in the pure Nash equilibrium $\mathbf s$ is 
	\begin{eqnarray*}
\EE[u_i(\mathbf s)]&=&\sum_{j\in [n]} p_{ij} (s_i, \mathbf{s_{-i}}) v_{ij} \ge \sum_{j\in [n]} p_{ij} (s'_i, \mathbf{s_{-i}}) v_{ij} \\ &\ge& \frac{1}{n} \sum_{j\in [n]} v_{ij} = \frac{1}{n},
	\end{eqnarray*} due to the fact that $\mathbf s$ is pure Nash equilibrium and $s'_i$ is a safe strategy of agent $i$. On the other hand, the utility of agent  $i\in I$ can be calculated by $\EE[u_i(\mathbf s)]=\sum_{j\in [n]} p_{ij} (s_i, s_{-i}) v_{ij} = (\sum_{j=1}^{k} p_{ij})/k$. Because $\mathbf{s}$ is a pure Nash equilibrium, it holds that $\EE[u_i] \ge 1/n$, so we get that $\sum_{j=1}^{k} p_{ij} \ge k/n$  for all $i\in I$.
As for the rest of the agents, \[\sum_{i\in N\backslash I} \sum_{j=1}^{k} p_{ij} = k - \sum_{i\in I} \sum_{j=1}^{k} p_{ij} \le k - (n-k) \frac{k}{n} = \frac{k^2}{n}.\]
This implies that the contribution to the social welfare from agents not in $I$ is at most $k^2/n$ and the expected social welfare of $M$ will be at most $1+(k^2/n)$. It holds that $SW_{OPT}(\mathbf{u}) \ge k$ and the bound follows by letting $k=\sqrt{n}$.
\end{proof}
Due to Theorem \ref{thm-PoS-lower}, in order to obtain an $\Omega(\sqrt{n})$ bound for a mechanism $M$, it suffices to prove that $M$ has a safe strategy. In fact, most reasonable mechanisms, including Random Priority and Probabilistic Serial, as well as all ordinal \emph{envy-free} mechanisms satisfy this property. 


\begin{definition}[Envy-freeness]
A mechanism $M$ is (ex-ante) \emph{envy-free} if for all agents $i$ and $r$ and all profiles $\mathbf{s}$, it holds that $\sum_{j=1}^n p_{ij}s_{ij}\geq \sum_{j=1}^n p_{rj}s_{rj}$. Furthermore, if $M$ is ordinal, then this implies $p_i^{M,\mathbf{s}} \succ_{s_i}^{sd} p_{r}^{M,\mathbf{s}}$.
\end{definition}	
 Given the interpretation of a truth-telling safe strategy as a ``proportionality'' property, the next lemma is not surprising.
 
\begin{lemma}\label{lem:envyfree}
Let $M$ be an ordinal, envy-free mechanism. Then for any agent $i$, the truth-telling strategy $u_i$ is a safe strategy.
\end{lemma}
 
\begin{proof}
Let $\mathbf{s} =(u_i,\mathbf{s_{-i}})$ be the strategy profile in which agent $i$ is truth-telling and the rest of the agents are playing some strategies $\mathbf{s_{-i}}$. Since $M$ is envy-free and ordinal, it holds that $\sum_{j=1}^{\ell} p^{\mathbf{s}}_{ij}\ge \sum_{j=1}^{\ell} p^{\mathbf{s}}_{rj}$ for all agents $r \in \{1,\ldots,n\}$ and all $\ell \in \{1,\ldots,n\}$. Summing up these inequalities for agents $r=1,2,\ldots,n$ we obtain
\begin{eqnarray*}
n\sum_{j=1}^\ell p^{\mathbf{s}}_{ij}&\ge& \sum_{j=1}^{\ell} \sum_{r=1}^n p^{\mathbf{s}}_{rj}=\ell,
\end{eqnarray*}
which implies that $\sum_{j=1}^\ell p^{\mathbf{s}}_{ij}\ge \frac{\ell}{n}$, for all $i \in \{1,\ldots,n\}$, and for all $\ell\in \{1,\ldots,n\}$.
\end{proof} 	
Note that since Probabilistic Serial is ordinal and envy-free \cite{BM:01}, by Lemma \ref{lem:envyfree}, it has a safe strategy and hence Theorem \ref{thm-PoS-lower} applies. It is not hard to see that Random Priority has a safe strategy too.

\begin{lemma}
Random Priority has a safe strategy.
\end{lemma}

\begin{proof}
Since Random Priority first fixes an ordering of agents uniformly at random, every agent $i$ has a probability of $1/n$ to be selected first to choose an item, a probability of $2/n$ to be selected first or second and so on. If the agent ranks her items truthfully, then for every $\ell=1,\ldots,n$, it holds that $\sum_{i=1}^{\ell} p_{ij} \geq \ell/n$.
\end{proof}	
In a sense, the safe strategy property is essential for the bound to hold; one can show that the \emph{randomly dictatorial} mechanism, that matches a uniformly chosen agent with her most preferred item and the rest of the agents with items based solely on that agent's reports achieves a constant Price of Stability. On the other hand, the Price of Anarchy of the mechanism is $\Omega(n)$. It would be interesting to show whether Price of Anarchy guarantees imply Price of Stability lower bounds in general.

\subsection{Unit-range representation} \label{sec:Unit-range}

Our second extension is concerned with the other normalization that is also
common in the literature~\cite{ZHOU:90,qiang,FFZ:14}, the unit-range representation,
that is, $\max_j u_i(j)=1$ and $\min_j u_i(j)=0$.  First, the Price of Anarchy guarantees from Section \ref{sec:PS} extend directly to the unit-range case. For Random Priority, since the results in \cite{FFZ:14} hold for this normalization as well, we can apply the same techniques to prove the bounds. For Probabilistic Serial, first, observe that Lemma \ref{lem:tv} holds independently of the representation. Secondly, in the proof of Theorem \ref{thm:PSPOA}, it now holds that 
\[ SW_{PS}(\mathbf{u,s}) \geq \frac{1}{n} \sum_{i=1}^n \sum_{j=1}^n u_{ij} \geq 1,\] 
which is sufficient for bounding the Price of Anarchy when $SW_{OPT}(\mathbf{u}) \leq \sqrt{n}$. Finally, the arguments for the case when $SW_{OPT}(\mathbf{u}) \leq \sqrt{n}$ hold for both representations. 

Concerning the lower bounds, we can prove the following theorem on the Price of Anarchy of deterministic mechanisms.

\begin{theorem}\label{thm:deterunitrange}
	The Price of Anarchy of any deterministic mechanism that always has pure Nash equilibria is $\Omega(n)$ for the unit-range representation.
\end{theorem}

Again, similarly to the corresponding bound in Section \ref{sec:Lower}, the mechanism that naively maximizes the sum of the reported valuations has pure Nash equilibria and achieves the above bound. The proof of Theorem \ref{thm:deterunitrange} is quite similar to that of Theorem \ref{thm:deterpoa} and is omitted due to lack of space.

More importantly, it is not clear whether the general lower bound on the
Price of Anarchy of all mechanisms that we proved in
Theorem~\ref{thm:loweranyunitsum} extends to the unit-range representation as well.
In fact, we do not know of any bound for the unit-range case and proving
one seems to be a quite complicated task.
As a first step in that direction, the following theorem obtains a lower
bound for $\epsilon$-{\em approximate} (pure) Nash equilibria.
A strategy profile is an $\epsilon$-approximate pure Nash equilibrium
if no agent can deviate to another strategy and improve her utility by
more than $\epsilon$. 
While the following result applies for any positive $\epsilon$, it is
weaker than a corresponding result for exact equilibria.

\begin{theorem}
Let $M$ be a mechanism and let $\epsilon\in(0,1)$.
The $\epsilon-$approximate Price of Anarchy of $M$ is $\Omega(n^{1/4})$
for the unit-range representation.
\end{theorem}

\begin{proof}
	Assume $n=k^2$, where $k\in \NN$ will be the
	size of a subset $I$ of ``important'' agents.
	We consider valuation profiles where, for some parameter $\delta\in(0,1)$,
	\begin{itemize}
		\item all agents have value 1 for item 1,
		\item there is a subset $I$ of agents with $|I|=k$ for which any agent $i\in I$ has value
		$\delta^2$ for any item $j\in\{2,\ldots,k+1\}$ and $0$ for all other items,
		\item for agent $i\not\in I$, $i$ has value $\delta^3$ for items $j\in\{2,\ldots,k+1\}$ and $0$
		for all other items.
	\end{itemize}
	Let $\mathbf{u}$ be such a valuation profile and let $\mathbf{s}$ be a Nash equilibrium. In the optimal allocation members of $I$ receive items $\{2,\ldots,k+1\}$ and such an allocation has social welfare $k\delta^2+1$.
	
	First, we claim that there are $k(1-2\delta)$ members of $I$ whose payoffs in $\mathbf{s}$ are at most $\delta$; call this set $X$. If that were false, then there would be more than
	$2k\delta$ members of $I$ whose payoffs in $\mathbf{s}$ were more than $\delta$.
	That would imply that the social welfare of $\mathbf{s}$ was more than $2k\delta^2$, which
	would contradict the optimal social welfare attainable, for large enough $n$
	(specifically, $n>1/\delta^4$).
	
	Next, we claim that there are at least $k(1-2\delta)$ non-members of $I$ whose
		probability (in $\mathbf{s}$) to receive any item in $\{1,\ldots,k+1\}$ is at most $4(k+1)/n$; call this set $Y$.
	To see this, observe that there are at least $\frac{3}{4}n$ agents who all have
	probability $\leq 4/n$ to receive item 1.
	Furthermore, there are at least $3n/4$ agents who all have probability
	$\leq 4k/n$ to receive an item from the set $2,\ldots,k+1$.
	Hence there are at least $n/2$ agents whose probabilities to obtain
	these items satisfy both properties.
	
	We now consider the operation of swapping the valuations of the agents in sets $X$ and $Y$ so that the members of $I$ from $X$ become non-members,
	and vice versa. We will argue that given that they were best-responding beforehand,
	they are $\delta$-best-responding afterwards. Consequently $\mathbf{s}$ is an
	$\delta$-NE of the modified set of agents.
	The optimum social welfare is unchanged by this operation since it only involves
	exchanging the payoff functions of pairs of agents.
	We show that the social welfare of $\mathbf{s}$ is some fraction of the optimal social welfare,
	that goes to $0$ as $n$ increases and $\delta$ decreases.
	
	Let $I'$ be the set of agents who, after the swap, have the higher utility
	of $\delta^2$ for getting items from $\{2,\ldots,k+1\}$. That is, $I'$ is the set of
	agents in $Y$, together with $I$ minus the agents in $X$.
	
	Following the above valuation swap, the agents in $X$ are $\delta$-best responding.
	To see this, note that these agents have had a reduction to their utilities for the outcome of receiving
	items from $\{2,\ldots,k+1\}$. This means that a profitable deviation for such
	agents should result in them being more likely to obtain item 1, in return for them
	being less likely to obtain an item from $\{2,\ldots,k+1\}$. However they cannot have
	probability more than $\delta$ to receive item 1, since that would contradict
	the property that their expected payoff was at most $\delta$.
	
	After the swap, the agents in $Y$ are also $\delta$-best responding.
	Again, these agents have had their utilities increased from $\delta^3$ to $\delta^2$
	for the outcome of receiving an item from $\{2,\ldots,k+1\}$.
	Hence any profitable deviation for such an agent would involve a reduction in
	the probability to get item 1 in return for an increased probability to get an
	item from $\{2,\ldots,k+1\}$. However, since the payoff for any item from
	$\{2,\ldots,k+1\}$ is only $\delta^2$, such a deviation pays less than $\delta$.
	
	Finally, observe that the social welfare of $\mathbf{s}$ under the new profile (after the swap)
		is at most $1+3k\delta^3$.
	To see this, note that (by an earlier argument and the definition of $I'$)
	$k(1-2\delta)$ members of $I'$ have probability
	at most $4(k+1)/n$ to receive any item from $\{1,\ldots,k+1\}$.
	To upper bound the expected social welfare, note that item 1 contributes 1 to the social welfare.
	Items in $\{2,\ldots,k+1\}$ contribute in total, $\delta^2$ times the expected number
	of members of $I'$ who get them, plus $\delta^3$ times the expected number
	of non-members of $I'$ who get them, which is at most
	$\delta^2k2\delta+\delta^3k(1-2\delta)$ which is less than $3k\delta^3$.
	
	Overall, the price of anarchy is at least $(k\delta^2+1)/3k\delta^3$, which
	is more than $1/\delta$. The statement of the theorem is obtained by choosing
	$\delta$ to be less than $\epsilon$, $n$ large enough for the arguments to hold
	for the chosen $\delta$, i.e. $n>1/\delta^4$.
\end{proof}

\section{Conclusion and Future Work}

Our results are rather negative: we identify a non-constant lower
bound on the Price of Anarchy for one-sided matching, and find
a matching upper bound achieved by well-known ordinal mechanisms.
However, such negative results are important to understand the
challenge faced by a social-welfare maximizer: for example,
we establish that it is not enough to elicit cardinal valuations,
in order to obtain good social welfare.
It may be that better welfare guarantees should use some assumption
of truth-bias, or some assumption of additional structure in agents'
preferences. 

An interesting direction of research would be to identify conditions 
on the valuation space that allow for constant values of the Price of Anarchy or impose
some distributional assumption on the inputs and quantify the average loss in welfare
due to selfish behavior. For the general, worst-case setting, one question raised is whether one can obtain
Price of Anarchy or Price of Stability bounds that match our upper bounds for the unit-range representation as well.


\clearpage

\bibliographystyle{plain}
\bibliography{refs}

\clearpage


\end{document}